\documentclass[12pt]{article}

\usepackage[T2A]{fontenc}
\usepackage[utf8]{inputenc}

\usepackage{amsmath,amsthm,amssymb,graphicx}
\newtheorem{lemma}{Lemma}
\newtheorem{theorem}{Theorem}
\theoremstyle{remark}
\newtheorem{remark}{Remark}
\DeclareMathOperator{\CT}{\textit{CT}}
\newcommand{\N}{\mathbb N}
\newcommand{\kv}{b_{\uparrow}}
\newcommand{\kh}{b_{\rightarrow}}
\newcommand{\kd}{b_{\nearrow}}
\newcommand{\mv}{a_{\uparrow}}
\newcommand{\mh}{a_{\rightarrow}}

\renewcommand{\le}{\leqslant}\renewcommand{\ge}{\geqslant}
\newcommand\es{\Lambda}

\title{On information content in certain objects}
\author{Nikolay Vereshchagin\\
Moscow State University, HSE University, Yandex%
\thanks{This research was supported by Russian Science Foundation, grant 20-11-20203, https://rscf.ru/en/project/20-11-20203/}
}
\date{}
\begin{document}
\maketitle
\begin{abstract}
The fine approach to measure information dependence is
based on the total conditional complexity $\CT(y|x)$, which is  defined  as the minimal length of a \emph{total} 
program that outputs $y$ on the  input $x$.
It is known that the total conditional complexity can be much larger than 
than the plain conditional complexity.
Such strings  $x,y$ are defined by means of a diagonal argument and are not otherwise interesting. 
In this paper we investigate whether this happens also for some natural objects.
More specifically, we consider the following  objects:
the number of strings of complexity less than $n$
and the lex first string of length $n$ and complexity $\ge n$.
It is known that they  have negligible mutual conditional complexities.
 In this paper we prove that their  mutual total conditional complexities
 may be large.
 This is the first example of  natural 
 objects whose plain conditional complexity is much less than the total one.
\end{abstract}


\section{Introduction}

The main notion of Algorithmic Information Theory is that of Kolmogorov
complexity $C(x)$ of a binary string  $x$ or, more generally, of a finite object $x$. 
It is defined as the minimal length of a program that outputs $x$ on the empty input $\es$,
assuming that our programming language is optimal in a sense.
Kolmogorov complexity of $x$ measures the amount of information present
in  $x$, regardless of whether that information is useful or not. 
For example, the string  $0^n$ consisting of $n$ zeros 
has at most $\log_2 n$ bits of information, while a random string of length  $n$ almost surely has 
about $n$ bits of information. 

Informally, the string $0^n$ contains the same information as  the number $n$.
How can we define formally what means that strings $x$ and $y$ contain  the same 
information? For this purpose we can use conditional Kolmogorov
complexity $C(x|y)$. It is defined  as the minimal length of a program that outputs $x$ on the  input $y$.
One can define that 
 $x$ and $y$ contain  the same 
information if both conditional complexities  $C(x|y),C(y|x)$ are negligible.
This definition agrees well with the fact that both conditional complexities  
 $C(0^n|n),C(n|0^n)$ are bounded by a constant.
 
However there exists a more fine approach to the notion of information.
It is based on the notion of the total conditional complexity $\CT(x|y)$, which 
is  defined  as the minimal length of a \emph{total} 
program that outputs $x$ on the  input $y$.
This approach was developed in~\cite{shen1,ver}.
It is known that the total conditional complexity can be much larger than 
than the plain
one, that is, there are strings  $x,y$ with $\CT(x|y)\gg C(x|y)$.  This fact was first observed in~\cite{shen2}. The strings  $x,y$ are defined by means of a diagonal argument and are not otherwise interesting. 

In this paper, we are interested what are mutual total conditional complexities  for some particular strings
that have natural definitions.
More specifically, for some  pairs of  strings  $x,y$ such that $C(x|y)\approx0$
we wonder whether we also have $\CT(x|y)\approx0$?
To this end, we have chosen 10 natural objects.
Namely, fix an optimal programming language $U$.
Consider then the following objects:
  \begin{enumerate}
\item
$B_n=\max\{m\in\N\mid C(m)<n\}$ --- the maximal natural number of Kolmogorov complexity less than $n$.
\item 
$BB_n=\max\{t_U(p,\Lambda)\mid |p|<n, U(p,\Lambda)\text{ is defined}\}$ --- the maximal
running time of a halting program of length less than $n$ on the empty input.
\item $T_n=\arg\max\{t_U(p,\Lambda)\mid  |p|<n, U(p,\Lambda)\text{ is defined}\}$ --- the most long-running halting program of
length less than $n$.
\item $LC_n=$ the list of all strings of Kolmogorov complexity less than $n$ together with their complexities. 
\item $L_n=$ the list of all strings of Kolmogorov complexity less than $n$. 
\item $\tilde L_n=$ the list of all programs of length less than $n$ (with respect to  $U$) that halt on the empty input.
 \item $G_n=$ the graph of the function $C(\cdot)$ on strings of length  $n$.
\item $H_n=$ the lex first string of length $n$ and complexity $\ge n$.
\item $N_n=|L_n|$ (the number of strings of complexity less than $n$).
\item $\tilde N_n=|\tilde L_n|$ (the number of programs of length less than $n$, with respect to  $U$, 
that halt on the empty input).
\end{enumerate}

When we talk about a list, we mean that strings of the list are
arranged in the alphabetical order. Item 3 
refers to the first such program, if there is more than one. 
Although not explicitly stated in the definition, we will assume that the number $n$ is the second component of each of these objects, 
so in fact each of them is a pair  (the object itself, the number $n$)\footnote{From 
most of these objects it is easy to find the number $n$, but from some it is not so easy. 
This issue will not be investigated in this paper.}.

Strictly speaking, all objects under consideration depend on the choice of the optimal programming language $U$, so in what follows we will use the 
notation $B^U_n,BB^U_n,\dots,\tilde N^{U}_n$.

These objects were discussed in the book~\cite{SUV}, except for the list $L_n$.%
\footnote{The book uses slightly different notation. Namely, the seventh object is denoted by $T_n$, 
and objects 3, 4, 6, 8, 9, 10 have no names at all.} 
In this book, it is proved that the complexity of all these objects is equal to $n$ up to a constant term. 
Moreover, it is proved there that, firstly, each of these objects for different 
$U$'s contains the same information, 
and secondly, different objects contain the same information. Both 
statements assume the rough approach to the definition of information, that is, 
assume that $x$ and $y$ have the same information if  both conditional complexities  $C(x|y),C(y|x)$ are negligible.

More precisely, the following holds 
\begin{theorem}[\cite{SUV}]\label{th1} 
For any objects $X,Y$ from the above list and any optimal programming languages $U,V$,
there are constants $c,d$ for which $C(X^U_n|Y^V_{n+c})\le d$ for all $n$. 
\end{theorem}
Informally speaking, the information contained in the object $X^U_n$
does not depend on the optimal programming language $U$, 
nor on the choice of the object $X$ from the above list. 
Note that the transition
from $n$ to $n+c$ in this theorem is natural, because Kolmogorov complexity can change by a constant when 
the programming language changes.

The proof of this theorem can be found in \cite[Theorem 15 on p. 25]{SUV}.
It has  the following minor point:
let $p$ be a program of length $d$ witnessing the inequality $C(X^U_n|Y^V_{n+c})\le d$. Then
the algorithm $U$ on input $(p,Y^V_{n+c})$ can run for a  very long time:
its running time may not be bounded by any total computable function of $Y^V_{n+c}$.
In other words, the proof does not establish that the total
complexity $\CT(X^U_n|Y^V_{n+c})$ is small.

A natural question arises whether it is possible to get rid of this disadvantage, that is,
are the quantities $\CT(X^U_n|Y^V_{n+c})$ also bounded by a constant?
Informally speaking, does the information in $X^U_n$ depend on the optimal programming language $U$ and
on selecting an object $X$ from the list, assuming the fine approach
to the definition of information?
We will split this question into the following two:
\begin{itemize}
\item For a given object $X$ from the above list and for any
  optimal programming languages $U,V$:
  is it true that $$\CT(X^U_n|X^V_{n+c})\le d$$ for some $c,d$ and all $n$
(does the information in $X^U_n$ depend on the choice of the optimal programming language)?
\item For given objects $X\ne Y$ from the list and for any
optimal programming language $U$: is it true that $\CT(X^U_n|Y^U_{n+c})\le d$ for some $c,d$ and all $n$
(does information depend on the object for a fixed optimal programming language)?
\end{itemize}
In total, we have
10 questions of the first type and 90 questions of the second type.
First, let us list all the questions where the answer is known.
Those answers are positive and are either obvious, or are essentially
obtained in~\cite{SUV} in the proof of the Theorem~\ref{th1}. That is,
the programs constructed in~\cite{SUV} to witness the inequalities $\CT(X^U_n|Y^V_{n+c})\le d$
are total.
To this end, we divide the objects into \emph{large}, 1, 2 and 4, and \emph{small} --- all other objects.

 \emph{The known answers to the first question.}
 For all large objects, the answer to the first question is positive:
 for any large $X$ and for any $U,V$,
 $\CT(X^U_n|X^V_{n+c})\le d$ for some $c,d$ and all $n$.
 We conjecture that for all small objects $X$ the answer to the first question is negative.

 \emph{The answers to the second question.}
 All large objects are equivalent to each other: 
 for any pair of large objects, the answer to the second question is positive.
 In addition, the answer is positive,
 if $X$ is a small object and $Y$ is a large one.
 Finally, the answer to the second question is positive (for obvious reasons)
 also for the following pairs (see Fig.~\ref{pic7}):
 \begin{itemize}
\item $X=N$, $Y=L$,
\item $X=\tilde N$, $Y=\tilde L$,
\item $X= H$, $Y=G$,
\item $X=H$, $Y=L$.
\end{itemize}
\begin{figure}[ht]
\begin{center}
\includegraphics{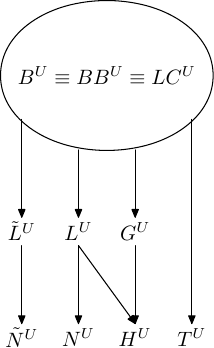}
\end{center}
\caption{Arrows show pairs of objects $(X,Y)$ for which
the answer to the second question is positive, that is, $\CT(X^U_n|Y^U_{n+c})\le d$ for some constants $c,d$.
The arrow leads from $Y$ to $X$. All three objects at the top are connected by bi-directional arrows. }\label{pic7}
\end{figure}
We conjecture that for the remaining pairs of objects $X,Y$ the answer to the second question is
negative.

In this paper, we have chosen 2 objects from the specified list,
$N_n$ and $H_n$, and found that for them the answers to all four questions
are negative (Theorems~\ref{th2}, \ref{th3}, \ref{th5} and \ref{th4}): all four total complexities
 $\CT(N^U_n|N^V_{n+c})$, $\CT(H^U_n|H^V_{n+c})$, $\CT(N^U_n|H^U_{ n+c})$ and $\CT(H^U_n|N^U_{n+c})$ can
grow
linearly as $n$ tends to infinity, for any $c\in\N$.
It remains an open question whether the last two statements (about the unlimited growth
of $\CT(N^U_n|H^U_{n+c})$ and $\CT(H^U_n|N^U_{n+c})$) hold
for \emph{every} optimal programming language $U$?

In the next section, we present the main definitions.
In Section~\ref{s3} we give the formulations of the results and their proofs.
The proofs use a game technique: for each
theorem, we define some two-person game with complete information and
we construct a computable winning strategy for one of the players. In Appendix~\ref{s4}
we make a complete analysis of a game of this type: we determine for which values of
parameters the first player has a winning strategy, and for which ones the second player does.

\section{Preliminaries}\label{s2}

We write \emph{string} to denote a  finite  binary  string.  Other
finite objects, such  as natural numbers or pairs of  strings,  may be encoded  into
strings  in  natural ways. The set  of  all strings is denoted  by
$\{0,1\}^*$ and the length of a string $x$ is denoted by $|x|$. The empty string is denoted by $\es$.

A \emph{programming language} is an algorithm 
$F$ from $\{0,1\}^*\times\{0,1\}^*$ to $\{0,1\}^*$.
The first argument of $F$ is called a program,
the second argument is called the input, and
$F(p,x)$ is called the output of program $p$ on input $x$. 
Let $t_F(p,x)$ denote the running time of $F$ on inputs $p,x$.
Let 
 $$
 C_F(y|x)=\min\{|p|\mid F(p,x)=y\}
 $$ 
 (conditional complexity
of $y$ relative to $x$ with respect to $F$).
A programming language $U$ is called \emph{optimal}
if for any other
programming language $F$ there exists a 
constant $d_F$
such that $C_U(y|x)\le C_F(y|x)+d_F$ for  all $x,y$.
By Solomonoff -- Kolmogorov theorem (see e.g. \cite{lv,SUV})
optimal programming languages
exist. 

Fix an optimal programming language $U$. We then define
\begin{itemize}
\item $C(y|x)=C_U(y|x)$ (conditional Kolmogorov
complexity of $y$ relative to $x$),
\item $C(x)=C(x|\es)$ (Kolmogorov complexity of $x$),
\end{itemize}

If $U(p,x)=y$, we say that $p$ is \emph{a  program for $y$ relative to $x$}.
If $U(p,\es)=x$, we say that \emph{$p$ is a program for $x$}.

To define total conditional complexity, we fix a programming language $T$
with the following property:
for any programming language $F$ there is a constant $c_F$
with
$$
\CT_T(y|x)\le \CT_F(y|x)+c_F
$$
for all strings $x,y$.
Here
$$
\CT_F(y|x)=\min\{|p|\mid F(p,x)=y, F(p,x')\text{ is defined for all strings }x'\}.
$$ 
The existence of such a programming language $T$ is proved exactly as Kolmogorov --- Solomonoff theorem.
We then define
$$\CT(y|x)=\CT_T(y|x)$$ 
(total conditional Kolmogorov
complexity of $y$ relative to $x$),


We use the following well known facts
\begin{itemize}
\item \emph{The upper graph} $\{(x,y,i)\mid C_F(y|x)< i\}$ of conditional Kolmogorov complexity
relative to any programming language $F$ 
is computably enumerable and for all $i,x$ there are less than $2^i$ different $y$'s
such that $C_U(y|x)< i$.
\item In particular, for any $l$ there is  a string of length $l$ with $C(x)\ge l$;
\item
Conversely, for any computably enumerable set $A$ such that  for all $i,x$ there are less than $2^i$ different $y$'s with $(x,y,i)\in A$ there is a constant $c$ such that
$C(y|x)\le \min\{i\mid (x,y,i)\in A\}+c$ for all $x,y$.
\end{itemize}

\section{Results}\label{s3}

\begin{theorem}\label{th2}
There are optimal programming languages $U,V$ such that
for all natural $c$ there are infinitely many $n$ with
$$\CT(N^U_n|N^V_{n+c})\ge s,$$ where $s=(n-5)/3$.
\end{theorem}
\begin{proof}
First, we decide for which $n$ we will prove the inequality
for the given $c$. To do this, we choose a sufficiently large constant $d$.
Then we fix any enumerable family $S$ of pairs of natural numbers
with the following properties:\\ (1) all the components of different pairs from $S$
differ by at least $d$,\\
(2) for each natural $c$, the family $S$ contains infinitely many pairs of the form $(n,n+c)$.\\
Such a family $S$ can be easily constructed by a greedy algorithm. More specifically, we arrange all
natural $c$ in a sequence in such a way that each $c$ appears  infinitely many times
in the sequence. For example, like this:
$$
001012012301234
\dots
$$
Then we consider all terms of the sequence in turn and 
for each $c$ we include in $S$ any pair    
$(n,n+c)$ that is at distance at least $d$ (in both coordinates) from every pair included in $S$ so far.

Let us fix any optimal  programming language $W$ and define both $U$ and $V$ on words that have   $e$ 
leading zeros, as $W$, i.e. $U(00\dots0p,x)=V( 00\dots0p,x)=W(p,x)$. 
Here $e$ is a sufficiently large constant to be chosen later.
Thus, we will ensure optimality of both $U$ and $V$. On the other hand, we are free
to define $U$ and $V$ on the remaining words. 
The fraction of such words is the closer to 1, the larger $e$ we choose. The other restriction is that 
both functions must be computable. 

We will ensure computability of $U,V$ by constructing  algorithms that enumerate their graphs.
To this end, let us start  enumerations of the graph of $W$,  of the graph of the programming language $T$ 
from the definition of total complexity, and of the set $S$. 
When a new pair $(p,x)$ appears in the graph of $W$, we include the pair
$(0^ep,x)$ in the graphs of $U$ and $V$. This may result in 
decreasing $C_U(x)$ and $C_V(x)$  
(initially both values are infinite). This means some numbers 
$N^U_n,N^V_{n}$ can increase by 1 time to time 
(initially they are equal to zero). 
More precisely, if the complexity of a word $x$ with respect to $U$ was equal to $i$ and at a certain step of enumeration
it becomes equal to $j<i$, then all numbers $N^U_n$ for $n=j+1,\dots ,i-1$ increase
by 1. The same happens with the numbers $N^V_{n}$. Since the number of programs of length less than $n$ starting with $e$ 
zeros is less than $2^{n-e}$, both numbers $N^U_n,N^V_{n}$ can increase less than $2^{n-e}$ times 
due to the enumeration of the graph of $W$.

We have to ensure that for any pair $(n,n+c)$ from $S$ the inequality 
$$
\CT(N^U_n|N^V_{n+c})\ge s
$$ 
holds. 
In other words,  there is no total program $p$ of length less than $s$
with $T(p,N^V_{n+c})=N^U_n$.

Let some pair $(n,n+c)$ from $S$ be given.  Note that there is 
no moment in the enumeration of  the graph of $T$ when we can be sure that  a
given program $p$ is total. However 
we do not need this, as we can even fool all programs $p$ for which the function $x\mapsto T(p,x)$ 
is defined on all words $x$ of length $n+c$. When such $p$ of length less than $s$ is discovered, 
we will say that all pairs of the form $(x,T(p,x))$, $|x|=n+c$, become red 
(initially we assume that all pairs are white). 
For each pair $(n,n+c)$ from $S$, our goal is to extend the functions $U$ and $V$ so that the pair $(N^V_{n+c},N^U_n)$ 
is white after the appearance of the last red pair for this $n$. To do this, 
we increase $N^U_n$ or $N^V_{n+c}$ by choosing a new word $y$ 
whose complexity is still large, and a new string $p$ of length $n-1$ or $n+c-1$, respectively, that 
does not start with $e$ zeros. Then we  set $U(p,\es)=y$ or $V(p,\es)=y$. We can do this $2^{n-1}(1-2^{-e})$ times for $N^U_n$ and $2^{n+c-1}(1-2^{-e})$ times for $N^V_{n+c}$.

When we increase $N^U_i$ ($N^V_{i}$), we also automatically increase the numbers $N^U_n$ ($N^V_{n}$) for all $n>i$ as well. 
We will imagine that these ``unwanted''
increments of the numbers $N^U_n$ and $N^V_{n}$ are made by an imaginary adversary. 
If the constant $d$ in the property (1) of the set $S$ is sufficiently large, then the number of unwanted 
increments of the number $N^U_n$ is much less than the number of useful increments. Indeed, it does not exceed 
$$
\sum_{i<n-d}2^i<2^{n-d}\ll 2^{n-1}. 
$$
The same is true for increments of $N^V_n$.

Increases of $N^U_n$ and $N^V_{n}$ caused by the
enumeration of the graph of $W$ will also be referred to as unwanted ones and assumed to be made by the 
adversary. There are less than $2^{n-e}$ of them, which is also small compared to $2^{n-1}$,
provided $e$ is large.

Thus, we are essentially playing   in parallel infinite number of games $G_{n,c}$, where  
$(n,n+c)\in S$, with an imaginary opponent. 
In the game $G_{n,c}$ we can increase $2^{n-1}(1-2^{-e})$ times the number $N^U_n$, and $2^{n+c-1}( 1-2^{-e})$ times the number $N^V_{n+c}$. The opponent can increase these numbers $2^{n-d}+2^{n-e}$ and $2^{n+c-d}+2^{n+c-e}$ times\footnote{and even $2^{n-d}+2^{n +c-e}$ times}, respectively. In addition, the opponent can sometimes declare some pairs of natural numbers as undesirable by coloring them red. She can make such a declaration $2^{s}$ times, indicating for each $i<2^{n+c}$ some $j<2^n$ for which the pair $(i,j)$ turns red. Our goal is to guarantee that after the opponent makes her last move, the pair $(N^U_n,N^V_{n+c})$ is white.

We will construct a winning strategy in each game $G_{n,c}$. To do this, we give a more general description of the game, 
independent of specific values of parameters.

\emph{The description of the game.} 
The game is determined
by the natural parameters 
$$
\mv,\mh,\kv,\kh,r,p
$$ 
and is played by two players, Alice and Bob, who move in turn, starting with Alice. 
The ``game board'' is the set $\N\times\N$, whose elements are  called \emph{cells}. 
At each moment of the play, there is a token on some cell, which at the beginning is on the cell $(0,0)$. 
On each her move, Alice either passes or moves the token to the right or up, that is, from the current cell $(i,j)$ 
to one of the cells $(i+1,j)$, $(i,j+1)$. On each his turn, Bob either passes or makes one or more moves of the same type. 
Alice can move the token up no more than $\mv$ times and to the right no more than $\mh$ times, 
and Bob can move the token no more than $\kv$ and $\kh$ times, respectively.

In addition, on any of her moves, 
Alice can declare any\footnote{We could restrict Alice by  letting her declare red only cells from the graph of some function. 
However, this restriction does not help to  construct a winning strategy.} 
at most $r$ cells \emph{red} (initially all cells are white). 
She can make such a declaration less than $p$ times. 
The game continues countably many moves, and at the end of the game the token 
will be located at some limit position 
(since each player can move it only a finite number of times). 
Alice is declared the winner if this cell is red, otherwise Bob. 
The following lemma provides a sufficient condition for the existence of Bob's winning strategy.

\begin{lemma}\label{l1}
In this game, Bob has a computable winning strategy if
$$
\kh/f-\mv\ge \sqrt{rp^3/f},\quad
\kv-\mv-\mh/f\ge 2\sqrt{rp^3/f}
$$
for some positive integer $f$. 
\end{lemma}
Recall that we are interested in the following regime: 
$p=2^s=2^{(n-5)/3}$ and $r=2^{n+c}$. 
Let $f=2^c$. With such parameters, the right-hand sides of the inequalities in the condition of the lemma are respectively $2^{n-2.5}$ and $2^{n-1.5}$, 
which is a constant times less than $2^{n-1}$. 
On the other hand, $\kv\approx 2^{n-1}$ and $\kh\approx 2^{n+c-1}$, thus $\kh/f\approx 2^{n-1}$.
The larger are constants $e,d$ the more precise these equalities are.
Besides   $\mv$ and $\mh/f$ are negigible compared to $2^{n-1}$ provided  $d,e$ are large.
Thus by choosing sufficiently large constants $d,e$, we can ensure that the left-hand sides of the inequalities are arbitrarily close to $2^{n-1}$, 
and hence ensure that both inequalities from the condition of the lemma are true. 
 So we can win all $G_{n,c}$ games. To complete the proof of the theorem, it remains to prove the Lemma~\ref{l1}.
\end{proof}

\begin{proof}[Proof of Lemma~\ref{l1}]
Before Alice has made the first announcement about red cells, we, playing as Bob, pass. 
Then after each such declaration, we do the following. Let the token currently be on the cell $(x,y)$. For each $j=0,1,\dots,\Delta-1$, where $\Delta$ is some parameter that will be  chosen later, consider the following set of cells
$$
A_j=\{(x+i,y+\lfloor i/f\rfloor+j)\mid i\ge0\}.
$$
These sets will henceforth be called \emph{stair cases}, since they consist of horizontal steps of length $f$. 
Among these sets, we choose the one with the least number of red cells. 
These sets are pairwise disjoint, so the chosen set contains at most $rp/\Delta$ red cells.
Then, moving up, we ensure that the token falls into the chosen stair case. 
Note that for this it suffices to make less than $\Delta$ moves. Indeed, now the token is in $A_0$. 
When moving up by $j$ steps, we get into the set $A_j$.

Then, until Alice makes the next announcement about red cells, 
we try to keep the token on a white cell of the chosen stair case. To this end we perform
the following actions.

\emph{Returning the token back to the selected stair case:}
if Alice has moved the token up, on the next move we shift it no more than $f$ times 
to the right in order to return the token to the chosen stair case. 
If Alice moves the token  to the right and before that the token was at the end of a step 
(and therefore left the stair case), we move it up on the next move. 

\emph{Avoiding red cells:}
if the token is on a red cell, then we shift it to the right, and if after that the token leaves the selected stair case, 
then we shift it up. If after these actions the token  is still on a red cell, then we repeat this until it is on a
white cell.

Our strategy has been described. It is obvious that it wins, because after each of our moves the token is on a white cell, 
which means that in the limit it will end up on a white cell. 
It remains to prove  that the strategy does not exceed the limit of shifts. 

We make four kinds of shifts:
\begin{itemize}
\item 
Shifts to the right in order to move the token away from a red cell. 
The number of
such shifts is less than the number of red cells in the chosen stair case multiplied by the number of red cell declarations, 
i.e., less than $(rp/\Delta)\cdot p$.
\item Shifts to the right to place the token in the chosen stair case after Alice has moved it up. 
There are less than $\mv\cdot f$ such shifts. 
\item Shifts up to place the token in the selected stair case after Alice's next announcement of red cells. 
There are less than $\Delta\cdot p$ such shifts. 
\item Shifts up to place the token in the selected stair case after it has moved to the right (by Alice or us) and left a step. 
Denote by $\kh'$ the actual number of shifts to the right that we have made. Then the total number of shifts to the right is at most $\kh'+\mh$, 
but only each $ f$th of them takes the token out of the stair case. 
Therefore, there are less than $(\kh'+\mh)/f$ such shifts. 
As we saw above, $\kh'<rp^2/\Delta + \mv f$. 
Therefore, there are less than
$$
(rp^2/\Delta +   \mv f  +\mh)/f=rp^2/\Delta f  +   \mv+ \mh/f
$$
such shifts.
\end{itemize}
So, in order not to exceed the limits, we need inequalities
\begin{align*}
\kh &\ge rp^2/\Delta +   \mv f,\\
\kv &\ge rp^2/\Delta f  +\Delta p +   \mv+ \mh/f.
\end{align*}
Now we select the $\Delta$ parameter by setting $\Delta=\sqrt{rp/f}$. With this choice, our inequalities turn into the inequalities
\begin{align*}
\kh &\ge \sqrt{rp^3f} +   \mv f,\\
\kv &\ge2\sqrt{rp^3/f} +   \mv  +\mh/f
\end{align*}
from the condition of the lemma.
\end{proof}

\begin{theorem}\label{th3}
There are optimal programming languages $U,V$ such that for all positive integers $c$ 
there are infinitely many $n$ with 
$$
\CT(H^U_n|H^V_{n+ c})\ge (n-5)/3.
$$
\end{theorem}
\begin{proof}
Surprisingly, the proof of this theorem is almost the same as the proof 
of the previous one. Now the moves of both players in the game $G_{n,c}$ consist in declaring some words of length $n$ or $n+c$ \emph{simple}, 
that is, declare that their complexities are less than their lengths. 
To do this, playing as Bob, we will provide  descriptions of length $n-1$ and $n+c-1$, respectively, for such words.

We will declare as simple only those words of a given length 
which at the current  moment are lex first non-simple ones. 
Alice is not required to act in the same way, that is, to declare as simple only the first non-simple words. 
However, without loss of generality, we can assume that she also keeps this rule. 
Indeed, if Alice declares simple a word $x$ of length $n$ that is not currently the smallest non-simple word, 
this declaration does not chnage $H_n$.
We can mentally remember this move and assume that Alice has passed. 
We will remember this move of hers when $x$ turns out to be the smallest non-simple word. 
Such a rearrangement of moves, obviously, will not change anything. 
If we assume that Alice either passes or declares the smallest of the non-simple words to be simple 
(or declares the next portion of red cells), 
then we get the same game $G_{n,c}$ as in the previous theorem.
\end{proof}

\begin{theorem}\label{th5}
There exists an optimal programming language $U$ such that for all positive integers $c$ there are infinitely many $n$ with
$$
\CT(H^U_n|N^U_{n+c})\ge (n-7)/3.
$$
\end{theorem}
\begin{proof}
The proof is similar to that of
the previous theorems. 
But there is also an important difference. 
We cannot increment the numbers $H^U_n$ and $N^U_{n+c}$ independently because they depend on the same programming language. 
Each increase of $H^U_n$ produces an increase of $N^U_i$ for all $i\ge n$.
If $i$ is a second component of a pair from $S$ different from the pair $(n,n+c)$,
then the number of such increases is negligible compared to $2^{i}$ due to property (1) of $S$,
and can be ignored. 
However   increases of  $N^U_{n+c}$ caused by increasing $H^U_n$ cannot be ignored. 
On the other hand, $N^U_{n+c}$ can be increased by simplifying words whose lengths do not lie in $S$. 
This will make it possible to make such increases without changing $H^U_n$ and, in general, 
all numbers of the form $H^U_i$. 
Thus, we need to build a winning strategy in a game where we are allowed horizontal $((i,j)\to (i+1,j)$) and diagonal 
$((i,j)\to (i+1,j+1)$) 
shifts. 
Alice has the same restriction, 
but it will be easier to assume that Alice is allowed both vertical and horizontal shifts 
(a diagonal shift can be implemented by two consecutive shifts, vertical and horizontal). 
A computable winning strategy is delivered by the following lemma.
 
\begin{lemma}\label{l2}
Consider the game, where Bob is allowed to make horizontal and diagonal shifts 
(no more than $\kh$ and $\kd$, respectively), 
Alice is allowed vertical and horizontal shifts (no more than $\mv $ and $\mh$), 
and  Alice can declare $p$ times  any 
at most $r$ cells red.
In this game Bob has a computable winning strategy if
$$
\kh/f-\mv\ge \sqrt{rp^3/f}
,\quad
\kd-2\mv-2\mh/f\ge 4\sqrt{rp^3/f}
$$
for some integer  $f>1$. 
\end{lemma}
Recall that $r=2^{n+c}$ and $p=2^{(n-7)/3}$. Let us apply Lemma~\ref{l2} to $f=2^{c+1}$. 
The inequalities of the lemma become
$$
\kh2^{-c-1}-\mv \ge 2^{n-3.5},\quad
\kd-2\mv-2^{-c}\mh\ge 2^{n-1.5}.
$$
By choosing sufficiently large $d,e$ we can make the left-hand sides be arbitrarily close to $2^{n-2}$ and $2^{n-1}$, respectively, 
and therefore larger than their right-hand sides. It remains to prove Lemma~\ref{l2}.
\end{proof}
 
\begin{proof}[Proof of Lemma~\ref{l2}]
The proof is similar to that of Lemma~\ref{l1}.
More precisely, let the token 
after the next announcement of Alice be on the cell $(x,y)$. 
Again consider sets of the form 
$$
A_j=\{x+i,y+\lfloor i/f \rfloor+j)\mid i\ge0\},
$$
where $j=0,1,\dots,\Delta-1$, consisting of horizontal steps of length $f\ge 2$. 
Among them, we choose the set $A_j$ with the smallest number of red cells. These sets are pairwise
disjoint, so the chosen set contains at most $rp/\Delta$ red cells.

Then we make as many diagonal shifts as needed to place the token into the selected set. 
Because the step length is at least 2, 
this is indeed possible, and it suffices to make less than $2\Delta$ shifts. 
Indeed, at the beginning we are in the set $A_0$ ($i=0$). 
With one diagonal shift, we get into $A_1$ ($i=1$). Making one more diagonal shift, we get into $A_2$ ($i=2$). 
After the $f$th diagonal shift, we will remain in the same set, but the next shift will move the token to the next set.

Then we keep the token in the selected set $A_j$ 
and move it away from red cells. 

\emph{Returning the token back to the selected stair case:}
If Alice moves the token up, on the next our move we make at most $f$ right shifts to return the token to the chosen set. 
If Alice moves the token to the right and before the move the token  was at the end of a step (and therefore has left $A_j$), 
we make  the diagonal shift, returning the token to $A_j$. 

\emph{Avoiding red cells:}
If the token is on a red cell, then we shift it to the right, and if after that the token left the selected stair case, 
then we make a diagonal shift. If after these shifts  the token is still on a red cell, then we repeat the moves until it is on a white cell.

The strategy has been described. 
It remains to prove that it does not violate the rules, that is, we do not exceed the limit of shifts. 
We make four kinds of shifts: 
\begin{itemize}
\item Shifts to the right  in order to move the token away from a red cell. The number of 
such shifts is less than the number of red cells in the selected set, 
multiplied by the number of red cells declarations, that is, less than $rp^2/\Delta$. 
\item Shifts to the right to place the token in the selected stair case 
after Alice shifts it up. 
There are less than $f \mv$ such shifts. 
\item Diagonal shifts to place the token in the selected stair case after Alice's announcements of red cells. 
There are less than $2\Delta\cdot p$ of such shifts. 
\item Diagonal shifts to place the token in the selected stair case after it is moved to the right (by Alice or us) and leaves the selected set. 
Denote by $\kh'$ the actual number of horizontal shifts that we have made. 
Then the total number of horizontal shifts is at most $\kh'+\mh$, but only each $f$th of them takes a token out of the stair case. 
Taking into account that by moving the token diagonally, we ourselves additionally increase the horizontal coordinate, 
we get that the number of such diagonal shifts is less than
$$
(\kh'+\mh)/f+ (\kh'+\mh)/f^2+\dots\le 2 (\kh'+\mh)/f.
$$
\end{itemize}
So, in order not to exceed the limits, we need inequalities
\begin{align*}
\kh&\ge rp^2/\Delta +   f \mv\\
\kd&\ge2(rp^2/\Delta +   f \mv +\mh)/f +2\cdot\Delta p.
\end{align*}
These are the same inequalities as in Lemma~\ref{l1}, except that the factors 2 appeared on the right side of the second inequality. 
Putting $\Delta=\sqrt{rp/f}$, we get exactly the inequalities from the condition of the lemma.
\end{proof}

\begin{theorem}\label{th4}
There is an optimal programming language $U$ such that for all positive integers $c$ there are infinitely many $n$ with
$$\CT(N^U_n|H^U_{n+c})\ge (n-7)/3.$$
\end{theorem}
\begin{proof}
We argue as before, that is, we choose a set $S$ with the same properties and ensure that the inequality 
in the condition holds for all pairs $(n,n+c)\in S$. Now
we can increase the numbers $N^U_n$ without changing the numbers $H^U_i$. 
For $c>0$ we can increase the number $H^U_{n+c}$ without changing the number  $N^U_{n}$. 
And the case $c=0$ is special --- by increasing the number $H^U_{n+c}$, we automatically increase the number $N^U_{n}$ as well. 
Therefore, in the game for $c=0$ we can make vertical and diagonal moves, 
and for the remaining $c$'s we can make vertical and horizontal moves. 
Therefore, for $c=0$ we need Lemma~\ref{l2}, and for the remaining $c$'s we need Lemma~\ref{l1}. 
The rest of the proof is similar.
\end{proof}

\section{Open problems}

1. Prove that  for any object  $X=\tilde L,L,G,\tilde N,T$ there are optimal programming languages $U,V$ such that
 $\CT(X^U_n|X^V_{n+c})> d$ for all $c,d$ and  for infinitely many $n$.
 
 2.  Prove that  for all pairs   $(X,Y)$ where 
\begin{itemize}
\item $X=B,N,L,G,H,T$,\quad $Y=\tilde L$,
\item $X=B,\tilde L,G,\tilde N,T$, \quad$Y=L$,
\item $X=B,\tilde L,L,\tilde N, N,T$, \quad$Y=G$,
\item $X$ is any object different from $Y$ and  $Y=\tilde N,N,H,T$,
\end{itemize}
 there is an optimal programming language $U$ such that
 $\CT(X^U_n|Y^U_{n+c})> d$ for all $c,d$ and  for infinitely many $n$.

3. Prove that the statement of the previous item holds for \emph{every} optimal programming language $U$. 

4. Provide a complete analysis for the games of Lemmas~1 and~2.

\section{Acknowledgments.}
The author is sincerely grateful to Dmitri Polushkin for the permission to publish his result (Theorem~\ref{th7})
and to all participants of Kolmogorov seminar at MSU.

\appendix

\section{A complete analysis of a similar game}\label{s4}

The proofs of our theorems were based on an analysis of two similar games between Alice and Bob. 
To prove the theorem, it was enough to construct a computable winning strategy for Bob. 
At the same time, we do not know exactly, 
for which values of the parameters Alice wins, and for which Bob wins. 
One of the reasons for this is the abundance of parameters 
(there are six of them). 
In this section, we give a complete analysis of a game of the first type 
(both players can move the token up and to the right), 
in which the total number of shifts of each player is limited, 
and  $p=1$ (Alice can only make one announcement about red cells). 
This game has three parameters $r,a,b$.

So, there are two players in the \emph{$(r,a,b)$-game}, Alice and Bob. 
Alice goes first and must choose 
at most  $r$-element subset $R$ of the set $\N\times\N$ whose cells we will call \emph{red}. 
In the course of the play a token moves from cell to cell, at the beginning the token is on the $(0,0)$ cell. 
If the token is on a red cell, then Bob moves the token to the right or up, that is, 
from the current cell $(i,j)$ to one of the cells $(i+1,j)$, $(i,j +1)$. 
If the token is on a white cell, then Alice moves it to the right or up. 
Alice can move the token at most $a$ times and Bob at most $b$ times. 
The game continues until one of the players cannot make the next move without exceeding the limit of moves --- 
this player is declared the loser. 
We may assume that the playing field consists of all cells of the triangle (see Fig.~\ref{pic1})
$$
T_{a+b}=\{(i,j)\mid i\ge 0, j\ge 0, i+j\le a+b\}.
$$ 
\begin{figure}[ht]
\begin{center}
\includegraphics{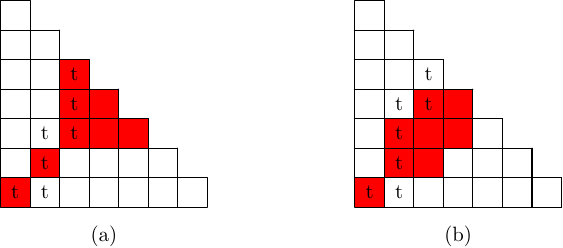}
\end{center}
\caption{The game board in the $(r=8,a=2,b=4)$-game. 
Cells from Alice's first move are marked in red. 
The letters ``t'' mark the path of the token in a  possible play won by Alice (a) 
and Bob (b). }\label{pic1}
\end{figure}
We will call the remaining cells \emph{unreachable in the $(r,a,b)$-game}. 
In addition, we will call a cell $(i,j)$ \emph{reachable from a cell $(k,l)$ in the $(r,a,b)$-game} 
if $(i,j)\in T_ {a+b}$ and $i\ge k,j\ge l$. 
The set of such cells will be denoted by $T_{a+b}(k,l)$ and called \emph{the triangle with vertex $(k,l)$ 
in the $(r,a,b)$-game}.

The following theorem provides a complete characterization of the values of the parameters 
for which Alice has a winning strategy and for which it is Bob. 
\begin{theorem}[D. Polushkin~\cite{p}]\label{th7}
If $b \ge r$ or 
\begin{equation}
b\ge a \text{ and } \frac{(b - a + 1)(b - a + 2)}2>r-a,\label{eq-ind}
\end{equation} 
then Bob has a winning strategy. 
For all other $a,b,r$, Alice has a winning strategy. 
\end{theorem}
\begin{proof}
Assume first
that 
$b\ge r$. 
Then Bob always moves the token up. 
There are at most $r$ red cells, and each of them can be visited only once, so Bob's strategy is winning. 

Assume that  $r>b$. If $a>b$, then  Alice wins as follows. 
Her first move is to color red all the cells 
$\{(i,i)\mid 0\le i<r\}$ on the diagonal and then she plays as follows. 
As soon as Bob moves a token off the diagonal, 
Alice returns it back on her next move. Since $a> b$, this is indeed possible. 
Since Bob has less than $r$ moves, Alice wins.

It remains to consider the case  $a\le b<r$. 
We need to prove that if inequality~\eqref{eq-ind} is true, 
then Bob has a winning strategy, and otherwise Alice. 
We will construct a winning strategy for one of the players by induction on $a$.

\emph{Basis of induction.} 
For $a=0$ Alice can only make the first move by choosing $r$ red cells, 
and Bob can move the token to any cell of the triangle $T_b$, to which there is a path through red cells. 
Therefore, Alice wins if and only if she makes all cells of this triangle red on her first move, 
which is possible iff $r$ is not less than the number of cells in the triangle. 
This means that for 
$$
r\ge \frac{(b+1)(b+2)}2
$$
Alice has a winning strategy, otherwise Bob.

\emph{Induction step.} 
We need the following 
\begin{lemma} 
If Alice has a winning strategy in the $(r,a,b)$-game, then she has it in  the $(r+1,a+ 1,b+1)$-game. 
Conversely, if Bob has a winning strategy in the $(r,a,b)$-game, 
then he also has it in the $(r+1,a+1,b+1)$-game. 
\end{lemma}
Note that the induction step easily follows  from the lemma. 
Indeed, when all parameters are increased by 1, 
the values $r-a$ and $b-a$ do not change, so the truth/falsity of the inequality~\eqref{eq-ind} 
is preserved.

\begin{proof}
The first assertion is easy. 
Indeed, assume that Alice has a winning strategy $S$ in the $(r,a,b)$-game and let  $R$
denote her first move. 
Assume that in the $(r,a,b)$-game, cells are numbered by positive integers, 
and in the $(r+1,a+1,b+1)$-game by all non-negative integers. 
Then Alice's first move in the $(r+1,a+1,b+1)$-game is the set $R\cup\{(0,0)\}$. 
On the next move, Bob will move the token off the $(0,0)$ cell, after which Alice will move the token to the $(1,1)$ cell and launch the $S$ strategy.

Assume now that 
Bob has a winning strategy $S$ in the $(r,a,b)$-game. 
We want to construct his winning strategy in the $(r+1,a+1,b+1)$-game. 
Let us start playing the $(r+1,a+1,b+1)$-game for Bob in parallel with the mental $(r,a,b)$-game, 
making the same moves in both games, except for the very first one. 
Denote by $R$ the first move made by Alice in the real game. 
The set $R$ consists of at most $r+1$ red cells. In the mental game we will make another first move $R'$. 
The set $R'$ depends on whether the set $R$ includes at least one triangle or not.

\emph{Case 1 (no-triangle-case): the set $R$ does not include any triangle of the form $T_{a+b+2}(i,j)$. }
This means that $R$ does not contain any cell with the sum of coordinates $a+b+2$. 
Then we remove from $ R$ all the cells located at the greatest distance from the origin.
The resulting set is $R'$ (see Fig.~\ref{pic2}).
\begin{figure}
\begin{center}
\includegraphics{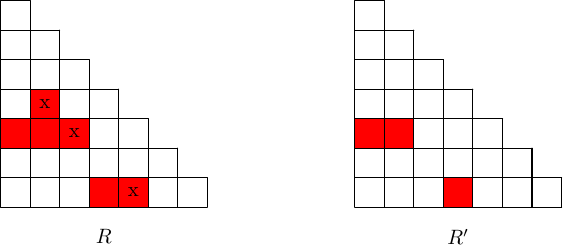}
\end{center}
\caption{ On the left, the set $R$ is shown, which does not include any triangle. 
On the right is the set $R'$ obtained from $R$ by removing the farthest cells, they are marked with the letter ``x''.}\label{pic2}
\end{figure}

\emph{Case 2 (triangle-case): the set $R$ includes at least one triangle}. 
Let us include in $R'$ all reachable (in the $r,a,b$-game) cells from $R$ 
and all reachable cells outside $R$ 
from which one can get in one move  inside some triangle $T_{a+b +2}(i,j)$ included in $R$. We will call the added cells \emph{pink} 
(see Fig.~\ref{pic3}).
\begin{figure}[ht]
\begin{center}
\includegraphics{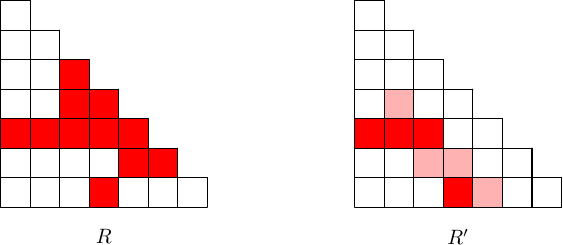}
\end{center}
\caption{On the left, a set $R$ including triangles is shown. 
On the right is the set $R'$ which includes all reachable cells from $R$, as well as reachable cells outside $R$, 
from which one can get  in one move  to some triangle $T_{a+b+2}(i, j)$ inside $R$ (drawn in pink).}\label{pic3}
\end{figure}
Later we will prove that $|R'|<|R|$.

Let us play $R'$ for Alice in the mental game, and then apply the $S$ strategy in the real and mental games. 
Alice's moves from the real game are translated into the mental game, 
and vice versa, Bob's moves from the mental game are translated into the real game. 
Only one of the following three events can stop such a  play: 
\begin{enumerate}
\item In the no-triangle-case, 
the token  hits a removed cell. 
Then Bob moves in the real game, 
and Alice moves in the mental game, 
so $S$ can no longer be used in the real game. 
\item In the triangle-case, 
the token hits a pink cell. 
Then Alice moves in the real game, 
and Bob moves in the mental game, 
so using $ S$ is no longer possible. 
\item The mental game is over (and won by Bob). 
\end{enumerate}
The second event is actually impossible, since in the triangle-case 
all the cells reachable in the mental game from pink cells are red or pink. 
Since $S$ is a winning strategy, it cannot at any time allow the token to land on a pink cell.

As soon as one of the remaining two events occurs, 
we stop the mental game and continue to play only the real game. 
Denote by $t$ the cell on which the token is located at that moment. 
Consider separately the no-triangle and triangle cases.
 
\emph{The no-triangle case. }
If the token ended up on a removed cell, 
then $t\in R$ and it is our turn to move. 
Then we move the token in either direction (we have at least one move left). 
The token will end up on a cell outside $R$, 
and no matter how far Alice moves 
the token, it will not return to $R$, since $t$ is a  farthest cell from $R$.

Otherwise (if the token is not on a removed cell) 
the third event happened and the mental game is won. That is, $t\notin R'$ and Alice has made $a$ moves. 
As $t$ was not removed from $R$, we also have $t\notin R$.
So Alice will make her last $(a+1)$st move by moving the token to some cell $t'$. 

Since Alice has exhausted her moves, 
we have enough moves to reach any cell of the  triangle $T_{a+b+2}(t')$. 
We just need to prove that this triangle is not included in $R$ --- then Bob can move the token to the nearest cell outside $R$. 
This is the case, since  $R$ does not include any triangle at all.

\emph{The triangle case.} 
 In this case, the mental game is won. That is,  the token is on some cell $t$ outside $R'$ and 
 Alice has made $a$ moves.
 Since $R'$ includes all cells from $R$ that can be reached in the mental game, the token is outside $R$, and Alice will 
 make her last $(a+1)$st move by shifting
 the token to some cell $t'$. 

Again Alice has exhausted her moves, hence
we have enough moves to reach any cell of the triangle $T_{a+b+2}(t')$. 
We claim that this triangle is not included in $R$. 
Indeed, if $T_{a+b+2}(t')\subset R$, then, since Alice moved a token from $t$ to $t'$ in one move, 
the cell $t$ is a pink cell in $R'$ , and we know that this is not the case
 (recall that $S$ has won the mental game).

It remains to prove that in the triangle case 
we have $|R'|<|R|$, that is, the number of unreachable cells in $R$ is greater than the number of pink cells. 
To this end, draw the diagonal $\{(i+t,j+t)\mid t\ge1\}$ from each pink cell $(i,j)$ (see Fig.~\ref{pic4}). 
\begin{figure}[ht]
\begin{center}
\includegraphics{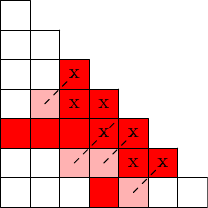}
\end{center}
\caption{Dashed diagonals define the embedding of the set of added cells in the set of deleted cells 
when converting $R$ to $R' $ in the triangle case. }
\label{pic4}
\end{figure} 
Some cell on this diagonal belongs to $R$ and is not reachable. 
Indeed, we know that the cell $(i,j)$ is adjacent to some triangle 
$T\subset R$, that is, one of the two cells $(i+1,j)$ or $(i,j+1) $ belongs to $T$. 
Consequently, all cells of the diagonal $(i+t,j+t)$ with $i+j+2t\le a+b+2$ also belong to $T$. 
Moreover, there are such cells, since the cell $(i,j)$ is reachable in the real game, 
so $i+j\le a+b$, which means $i+j+2\le a+b+2$. 
Depending on the parity of $i+j$, the sum of the coordinates of one of the cells 
of this diagonal is equal to $a+b+1$ or $a+b+2$, so it is not reachable and belongs to $T\subset R$. 
For different added cells $(i,j)$, the diagonals do not intersect. Indeed, otherwise one of the two cells would lie on the diagonal of the other, 
and hence would belong to $T$ and hence to $R$.

So, we have proved that there are at least as many pink cells as unreachable cells from $R$. 
It remains to prove that there are strictly more of them. 
Consider the vertex $(i_0,j_0)$ of any maximal triangle $T$ in $R$. 
Some cell of the diagonal $(i_0+t,j_0+t)$ has the sum of coordinates $a+b+1$ or $a+b+2$, 
which means that it is unreachable, 
but it does not lie on the diagonal of any pink cell. 
\end{proof}
The lemma is proved, and so is the theorem.
\end{proof}

\begin{remark}
It is useful to note that the proof of the theorem is constructive in the sense that it gives a description 
of the winning strategy of the player who has it. 
In more detail: Alice's winning strategy for the values of the parameters for which it exists is quite simple. 
For $a\le b$ it was explicitly described in the proof. 
For $a\ge b$, on her first move, she builds a ``kite'', that is, a triangle with a diagonal attached to its top 
(see Fig.~\ref{pic1}(a)). 
Then she spends all her moves returning the token to the tail of the kite when Bob takes it away from there.

Bob's winning strategy for $b\ge r$ is also simple: move the token up (say) when it lands on a red cell. 
His winning strategy for $r>b$ is not that simple. However,
from the proof of the lemma, it is easy to extract a recursive program that implements that strategy.
\end{remark}
\end{document}